\documentclass[11pt]{article}

\usepackage[margin=1in]{geometry}
\usepackage{amsmath,amssymb,amsthm,mathtools}
\usepackage{thmtools}
\usepackage{hyperref}
\usepackage{enumitem}
\usepackage{microtype}
\usepackage[capitalize]{cleveref}

\usepackage{bbm}
\hypersetup{
  colorlinks=true,
  linkcolor=blue,
  citecolor=blue,
  urlcolor=blue
}

\newtheorem{theorem}{Theorem}

\newtheorem{lemma}[theorem]{Lemma}
\newtheorem{definition}[theorem]{Definition}
\newtheorem{claim}[theorem]{Claim}
\newtheorem{corollary}[theorem]{Corollary}

\newtheorem{fact}[theorem]{Fact}
\theoremstyle{remark}

\newcommand{\R}{\mathbb{R}}

\newcommand{\Z}{\mathbb{Z}}

\newcommand{\ket}[1]{\lvert #1\rangle}
\newcommand{\bra}[1]{\langle #1\rvert}
\newcommand{\braket}[2]{\langle #1\vert #2\rangle}

\usepackage{crossreftools}

\usepackage{silence}
\makeatletter
\AddToHook{cmd/appendix/before}{\def\cref@section@alias{appendix}\def\cref@subsection@alias{appendix}}
\makeatother

\usepackage[backend=biber, style=numeric,  sorting=none, maxbibnames=99]{biblatex}
\allowdisplaybreaks
\title{The QAOA on the ring of disagrees}
\author{Kunal Marwaha\thanks{Google Quantum AI} \thanks{University of Chicago}
}
\date{}

\begin{document}
\maketitle
\begin{abstract}
    We study the performance of symmetric local algorithms finding large cuts on the cycle graph. Such algorithms that cannot see the whole graph at depth $p$ cut at most a $\frac{2p+1}{2p+2}$ fraction of edges in expectation.
    We prove that the QAOA achieves this value, a long-standing conjecture of~\cite{Farhi2014QAOA}.
    Curiously we do this without finding the optimal parameters. 
    Instead we show it is equivalent to find an optimal pair of Laurent polynomials of degree at most $2p-1$.
    This is made possible by recasting the QAOA on one qubit in the language of quantum signal processing.
\end{abstract}
\section{Introduction}
Let $C_{n}$ be the cycle with $n \ge 3$ vertices. The maximum cut of $C_n$ has size either $n$ or $n-1$ depending on whether $n$ is even or odd. It is easy to find such a cut. 
However not every algorithm is guaranteed to find a cut of this size. For example consider classical and quantum algorithms that are \textit{symmetric} (which means the circuit commutes with flipping all bits)
and \textit{local}. 
In these algorithms only the local neighborhood around an edge is used to decide whether or not an edge is cut.
Consider any such algorithm that observes local neighborhoods up to depth $p$. If the algorithm cannot see the whole graph $C_n$, it produces cuts of size at most $n \cdot \frac{2p+1}{2p+2}$~\cite{mbeng2019quantum,Bravyi2020Obstacles} in expectation.

The Quantum Approximate Optimization Algorithm (QAOA)~\cite{Farhi2014QAOA} is a symmetric local algorithm.
Here the problem of finding large cuts of $C_n$ has been called \textit{the ring of disagrees}.
For this problem \cite{Farhi2014QAOA} conjectured that the depth-$p$ QAOA finds cuts of size $n \cdot \frac{2p+1}{2p+2}$ in expectation when it cannot see the whole graph, and of maximum size otherwise.
While this conjecture is supported by significant numerical evidence~\cite{Farhi2014QAOA, Wang2018Fermionic, mbeng2019quantum,Rabinovich2025Overparametrization}, it remarkably has remained open for 12 years. In this work we finally prove the conjecture.

\subsection{Statement}
Define the \emph{cost} and \emph{mixer} Hamiltonians
\begin{align*}
    H_C &:= \sum_{i=1}^{n} \frac{1}{2} (I-Z_i Z_{i+1})\,, \\
    H_M &:= \sum_{i=1}^{n} X_i\,,
\end{align*}
where indices are modulo $n$. The depth-$p$ QAOA state, given $\vec{\beta} := (\beta_1, \dots, \beta_p), \vec{\gamma} :=(\gamma_1, \dots, \gamma_p)$, is
$$
\ket{\vec{\gamma}, \vec{\beta}} := e^{-i \beta_p H_M}e^{-i \gamma_p H_C} \dots e^{-i \beta_1 H_M}e^{-i \gamma_1 H_C} \ket{+}^{\otimes n}\,.
$$
Here we prove the following exact formulae:
\begin{theorem}[The QAOA cannot see the whole graph]
\label{thm:ring}
When $n \ge 2p + 2$,
    \begin{align*}
    \frac{1}{n} \cdot \max_{\vec{\gamma}, \vec{\beta}} \bra{\vec{\gamma}, \vec{\beta}} H_C \ket{\vec{\gamma}, \vec{\beta}} = \frac{2p+1}{2p+2}\,.
    \end{align*}
\end{theorem}
\begin{theorem}[The QAOA sees the whole graph]
\label{thm:ring2}
When $n < 2p+2$,
    \begin{align*}
    \frac{1}{n} \cdot \max_{\vec{\gamma}, \vec{\beta}} \bra{\vec{\gamma}, \vec{\beta}} H_C \ket{\vec{\gamma}, \vec{\beta}} = \begin{cases}
        1 \quad &n \textnormal{ is even} \\
        \frac{n-1}{n} \quad &n \textnormal{ is odd}\,.
    \end{cases}
    \end{align*}
\end{theorem}

\subsection{Proof outline}
The proof of \Cref{thm:ring} proceeds in three steps.

First, we use the free fermionic picture proven in~\cite{Wang2018Fermionic}. This shows that running the QAOA on $C_{2k}$ is equivalent to simultaneously running the QAOA on $k$ one-qubit systems. The angles are the same up to a factor of two and sign. Each one-qubit system uses a cost operator $X \cos \theta+ Z\sin \theta $, where system $j$ sets $\theta_j := \pi \cdot \frac{2j-1}{2k}$.

Second, we analyze the depth-$p$ QAOA on one qubit with cost operator $X \cos \theta + Z \sin \theta$. 
    We prove an \emph{expressibility theorem} which shows how the performance of the QAOA can be expressed as a function of $\theta$. With larger values of $p$ more complicated functions can be expressed. We show that choosing the angles $\vec{\gamma}, \vec{\beta}$ is equivalent to choosing a pair of Laurent polynomials of $z := e^{i \theta/2}$ whose degree is $\le 2p-1$. The performance of the QAOA is a fixed function of this pair of polynomials. Now we can look for optimal polynomial pairs instead of optimal angles.

Third, we suppose the QAOA cannot see the whole graph. This happens when $p < k$. On one qubit, the inner product of the final QAOA state with the $+1$-eigenstate of the cost operator has the form $\sum_{j=-p}^{p+1} c_{2j-1} z^{2j-1}$. The performance ratio of the QAOA simplifies to
    $$
    1 - \sum_{j=-p}^{p+1} |c_{2j-1}|^2\,.
    $$
    Using Cauchy-Schwarz, the sum is no smaller than $\frac{1}{2p+2}$. This bound is saturated by making all coefficients equal. We find a polynomial pair that achieves this condition. Curiously the optimal angles are not easy to write down.

So far we only have looked at cycles with even length. But the performance of the QAOA only depends on the local neighborhood, which is a path graph of $2p+2$ vertices. As a consequence this argument extends to cycles with odd length.

\vspace{-2.5em}
$$
$$
The proof of \Cref{thm:ring2} depends on whether $n$ is even or odd.

When $n$ is even we follow the proof  of \Cref{thm:ring} where $n = 2k$. 
Then $n < 2p + 2$ when $p \ge k$. 
We find a polynomial pair where the QAOA on a one-qubit system has error proportional to $\cos^2 (p \theta) \cos^2(\theta/2)$. Suppose $p = k$. Then for every one-qubit system $j$, $\cos^2(p \theta_j) = \cos^2(j \pi -\frac{\pi}{2}) = 0$. So the QAOA perfectly optimizes each one-qubit system. Then the QAOA perfectly optimizes $C_{2k}$ at $p = k$, and thus at any larger $p$. In this case it is easy to write down the optimal angles.

When $n$ is odd there is also a free fermionic picture proven in~\cite{Wang2018Fermionic} with slightly different angles. If $n = 2k + 1$ there are $k$ nontrivial free fermionic systems. As in the even case, we use an \textit{expressibility theorem} to look for a good pair of Laurent polynomials of $z := e^{i \theta/2}$ whose degree is $\le 2p- 1$. When $p = k$ we show the existence of a polynomial pair that optimizes every free fermionic system. Then the QAOA also optimizes $C_{2k+1}$ at any larger $p$. In contrast to the even case we did not find a simple closed form for the optimal angles.

\section{Reduction to one-qubit systems}
The following statement is from~\cite{Wang2018Fermionic}. For convenience, we provide a self-contained proof in \Cref{apx:fermion}.
\begin{claim}[\cite{Wang2018Fermionic}]
\label{claim:fermion}
Consider the one-qubit Hamiltonian
$$
h_c(\theta) = X \cos\theta + Z \sin\theta
$$
and one-qubit QAOA state
$$
\ket{\chi(\theta)} := e^{-i \tilde{\beta}_p X}e^{-i \tilde{\gamma_p} h_c(\theta)} \dots e^{-i \tilde{\beta}_1 X}e^{-i \tilde{\gamma_1} h_c(\theta)} \ket{+}\,,
$$
where
$\vec{\tilde{\gamma}} := (\tilde{\gamma}_1, \dots, \tilde{\gamma}_p) = (-\gamma_1, \dots, -\gamma_p)$ and $\vec{\tilde{\beta}} := (\tilde{\beta}_1, \dots, \tilde{\beta}_p) = (2\beta_1, \dots, 2\beta_p)$.
Then for $n = 2k$,
\begin{align*}
&\bra{\vec{\gamma}, \vec{\beta}} H_C \ket{\vec{\gamma}, \vec{\beta}} = k - \sum_{j=1}^k \bra{\chi(\theta_j)} h_c(\theta_j) \ket{\chi(\theta_j)} \,,\qquad \theta_j = \pi \cdot \frac{2j-1}{2k}\,.
\end{align*}
\end{claim}
\Cref{claim:fermion} allows us to understand the ring of disagrees by studying the QAOA on a single qubit.

\section{Expressibility of QAOA on one qubit}
Consider the QAOA on one qubit with mixer Hamiltonian $X$, and cost Hamiltonian
$$
h_c(\theta) = X \cos \theta + Z \sin \theta\,.
$$
Let $\ket{\theta}$ and $\ket{\theta_\perp}$ be the $+1$- and $-1$-eigenstates of $h_c(\theta)$. The performance of the QAOA here is
$$
\bra{\chi(\theta)}h_c(\theta)\ket{\chi(\theta)} = 
|\braket{\theta}{\chi(\theta)}|^2 - |\braket{\theta_\perp}{\chi(\theta)}|^2 = 2|\braket{\theta}{\chi(\theta)}|^2 - 1\,.
$$
We show how this value can be expressed as a function of $\theta$. We use a technique from the theory of quantum signal processing. This requires a bit of setup.

\begin{definition}
    A Laurent polynomial is a polynomial in both $z$ and $z^{-1}$, i.e. $f(z) = \sum_{j = - \infty}^{\infty} c_j z^j$ where finitely many $c_j$ are nonzero. Its degree is the maximum $r$ where either $c_r$ or $c_{-r}$ is nonzero.
\end{definition}
We use the notation $\mathbb{R}[z,z^{-1}]$ for Laurent polynomials with real coefficients. In general $z$ can be on the complex plane. We also use the following notation:
$$
f^\#(z) := f(z^{-1})
$$
When $f \in \mathbb{R}[z,z^{-1}]$ and $|z| = 1$, $f^\#(z) = f(z^*) = \left( f(z)\right)^*$.

For each integer $r \ge 0$, we define a special set of \textit{pairs} of Laurent polynomials $(A,B)$:

\begin{definition}
For each integer $r \ge 0$, let $\mathcal{S}_r$ be
\begin{equation}
\begin{aligned}
  \mathcal{S}_r:=\{(A,B):{}& A,B\in\R[z,z^{-1}], \\
  &\deg A,\deg B \le r,\\
  &\forall m \in 2\Z + 1, z^{m-r} A(z) \textnormal{ and } z^{m-r} B(z)\textnormal{ have no constant term}, \\
  &A(z)A^\#(z)+B(z)B^\#(z)=1\}.
\end{aligned}
\label{eq:S-n}
\end{equation}
\end{definition}
The third condition of $\mathcal{S}_r$ specifies that the only nonzero coefficients of $A$ and of $B$ are $c_j$ where $j$ has the same parity as $r$.

The performance of the QAOA on one qubit is characterized by a polynomial pair in $\mathcal{S}_r$:
\begin{theorem}[Expressibility theorem]
\label{thm:express}
    Fix $p \ge 1$. The depth-$p$ QAOA on one qubit, with cost function $h_c(\theta)$ and angles $\vec{\tilde{\gamma}} = (\tilde{\gamma}_1, \dots, \tilde{\gamma}_p), \vec{\tilde{\beta}} = (\tilde{\beta}_1, \dots, \tilde{\beta}_p)$, has performance
    \begin{align*}
    \bra{\chi(\theta)}h_c(\theta)\ket{\chi(\theta)} &=  -1 + \frac{1}{2}\left| z^{-2} A(z)  + z^2 A^\#(z) - iB(z)  - i B^\#(z) \right|^2\,, \qquad z = e^{i \theta/2}\,,
    \end{align*}
    for some pair of polynomials $(A,B) \in \mathcal{S}_{2p-1}$. Conversely, every $(A,B) \in \mathcal{S}_{2p-1}$ corresponds to angles  $\vec{\tilde{\gamma}} = (\tilde{\gamma}_1, \dots, \tilde{\gamma}_p), \vec{\tilde{\beta}} = (\tilde{\beta}_1, \dots, \tilde{\beta}_p)$ where the depth-$p$ QAOA achieves the above performance.
\end{theorem}
We will use the rest of this section to prove \Cref{thm:express}. 
\subsection{Rewriting the QAOA unitary}
We start by rewriting the QAOA to look more like quantum signal processing. For any Pauli $P$ and angle $\phi$, let
$$
R_P(\phi) := e^{-i \phi P}\,.
$$
\begin{claim}
The unitary associated with the depth-$p$ QAOA on one qubit, with cost function $h_c(\theta)$ and angles $\vec{\tilde{\gamma}} = (\tilde{\gamma}_1, \dots, \tilde{\gamma}_p), \vec{\tilde{\beta}} = (\tilde{\beta}_1, \dots, \tilde{\beta}_p)$, can be written as
\begin{align*}
    R_X(\phi_{2p}) R_Y(\theta/2) R_X(\phi_{2p-1}) R_Y(\theta/2) \dots R_X(\phi_{1}) R_Y(\theta/2)\,,
\end{align*}
where $\phi_{2j} = \tilde{\beta}_j + \frac{\pi}{2}$, $\phi_{2j-1} = \tilde{\gamma}_j - \frac{\pi}{2}$.
\end{claim}
\begin{proof}
We study one layer of the QAOA unitary:
$$
e^{-i \tilde{\beta} X} e^{-i \tilde{\gamma} h_c(\theta)} = R_X(\tilde{\beta}) e^{-i \tilde{\gamma} h_c(\theta)}
$$
We rewrite the rotation around $h_c(\theta)$ as
\begin{align*}
    h_c(\theta) = 
\begin{pmatrix}
    \sin \theta & \cos \theta \\
    \cos \theta & - \sin \theta
\end{pmatrix}
=
\begin{pmatrix}
    \cos \frac{\theta}{2} & \sin\frac{\theta}{2} \\
    -\sin \frac{\theta}{2}& \cos\frac{\theta}{2}
\end{pmatrix}
\begin{pmatrix}
    0 & 1 \\ 
    1 & 0
\end{pmatrix}
\begin{pmatrix}
    \cos \frac{\theta}{2} & -\sin \frac{\theta}{2} \\
    \sin \frac{\theta}{2} & \cos \frac{\theta}{2}
\end{pmatrix}
=
    R_Y(-\theta/2) X R_Y(\theta/2)
\end{align*}
\begin{align*}
    e^{-i \tilde{\gamma} h_c(\theta)} = R_Y(-\theta/2) R_X(\tilde{\gamma}) R_Y(\theta/2)
\end{align*}
Then one layer of the QAOA unitary has the form
$$
R_X(\tilde{\beta}) R_Y(-\theta/2) R_X(\tilde{\gamma}) R_Y(\theta/2)\,.
$$
We rewrite 
$$
R_Y(-\theta/2) = (-iX) R_Y(\theta/2) (+iX) =  R_X(\pi/2) R_Y(\theta/2) R_X(-\pi/2) \,.
$$
Then one layer of the QAOA unitary has the form
$$
R_X(\tilde{\beta}) \Big( R_X(\pi/2) R_Y(\theta/2) R_X(-\pi/2) \Big) R_X(\tilde{\gamma}) R_Y(\theta/2)
=
R_X(\tilde{\beta} + \pi/2) R_Y(\theta/2) R_X(\tilde{\gamma} - \pi/2) R_Y(\theta/2)\,.
$$
The claim follows after expanding each layer of the QAOA unitary in this form.
\end{proof}
\subsection{QAOA unitary from Laurent polynomials}
We show how the matrix elements of the QAOA unitary are related to Laurent polynomials. Let 
$$
z = e^{i \theta/2}\,.
$$
We write $R_X(\phi)$ and $R_Y(\theta/2)$ in the $Y$-basis (explicitly, mapping $X$ to $Y$, $Y$ to $Z$, and $Z$ to $X$):
\begin{align*}
    C(\phi) &:= R_X(\phi) = \begin{pmatrix}
    \cos(\phi) & - \sin(\phi) \\
    \sin(\phi) & \cos(\phi) 
\end{pmatrix}\,, \\
D &:= R_Y(\theta/2) = \begin{pmatrix} z^{-1} & 0 \\ 0 & z \end{pmatrix}\,.
\end{align*}
We start by considering matrices of the form $C(\phi_r) D C(\phi_{r-1}) \dots D C(\phi_0)$. When $r = 2p-1$, this matches a depth-$p$ QAOA unitary without the last $D$ unitary.

In this $Y$-basis, we show how the matrix entries are described by real Laurent polynomials in $\mathcal{S}_r$. This comes directly from the study of quantum signal processing~\cite{LowChuang2017,Haah2019Product}. We state the claims here and defer the proofs to \Cref{apx:qsp}:

\begin{claim}
\label{thm:state-theorem}
For every $r\ge0$,
\begin{equation*}
\left\{
  C(\phi_r)D C(\phi_{r-1})D\cdots D C(\phi_0)
  \begin{pmatrix}1\\0\end{pmatrix}
  :\phi_0,\dots,\phi_r\in\R
\right\}
=
\left\{
  \begin{pmatrix}A(z)\\B(z)\end{pmatrix}
  :(A,B)\in\mathcal{S}_r
\right\}\,.
\end{equation*}
\end{claim}
\begin{corollary}
\label{cor:state-theorem}
    For every $p\ge 1$, and in the $Y$-basis,
\begin{equation*}
\bigg\{
  C(\phi_{2p})D C(\phi_{2p-1})D\cdots D C(\phi_1) D
  :\phi_1,\dots,\phi_{2p}\in\R
\bigg\}
=
\left\{
  \begin{pmatrix}z^{-1} A(z) & -zB^\#(z) \\z^{-1} B(z) & z A^\#(z)\end{pmatrix}
  :(A,B)\in\mathcal{S}_{2p-1}
\right\}\,.
\end{equation*}
\end{corollary}
\subsection{Putting it together}
We now have the tools to prove the expressibility theorem.
\begin{proof}[Proof of \Cref{thm:express}]
  Fixing $\vec{\tilde{\gamma}},\vec{\tilde{\beta}}$ is equivalent to fixing $\phi_1, \dots, \phi_{2p}$ to some real values. We study
$$
\braket{\theta}{\chi(\theta)} = \bra{\theta}  C(\phi_{2p})D C(\phi_{2p-1})D\cdots D C(\phi_1) D \ket{+}\,.
$$
Recall that $h_c(\theta) = R_Y(-\theta/2) X R_Y(\theta/2)$. So $\ket{\theta} = R_Y(-\theta/2) \ket{+}$, and so
$$
\braket{\theta}{\chi(\theta)}
=
\bra{+} D  C(\phi_{2p})D C(\phi_{2p-1})D\cdots D C(\phi_1) D \ket{+}\,.
$$
By \Cref{cor:state-theorem}, the matrix below  can be equivalently written with some pair $(A,B) \in \mathcal{S}_{2p-1}$:
$$
 D C(\phi_{2p})D C(\phi_{2p-1})D\cdots D C(\phi_1) D = \begin{pmatrix}
   z^{-1}  & 0 \\
   0  & z
\end{pmatrix}
\begin{pmatrix}
   z^{-1} A(z)  & -z B^\#(z) \\
   z^{-1} B(z)  & z A^\#(z)
\end{pmatrix}
=
\begin{pmatrix}
   z^{-2} A(z) & -B^\#(z) \\
   B(z)  & z^2 A^\#(z)
\end{pmatrix}
$$
Writing $\ket{+}$ in the $Y$-basis, we have $\ket{+ } = \frac{1}{\sqrt{2}}\begin{pmatrix} 1 \\  i\end{pmatrix}$. For this pair $(A,B) \in \mathcal{S}_{2p-1}$,
$$
\braket{\theta}{\chi(\theta)} = \frac{1}{2} \left( z^{-2} A(z)  + z^2 A^\#(z) - iB(z)  - i B^\#(z) \right)\,,
$$
Then the performance of the depth-$p$ QAOA has the form
\begin{align*}
\bra{\chi(\theta)} h_c(\theta) \ket{\chi(\theta)} =  2|\braket{\theta}{\chi(\theta)}|^2 - 1
=
-1+\frac{1}{2}\left| z^{-2} A(z)  + z^2 A^\#(z) - iB(z)  - i B^\#(z) \right|^2 
\,.
\end{align*}
The converse direction holds because \Cref{cor:state-theorem} is an equivalence.
\end{proof}
\Cref{thm:express} means that for the QAOA on the ring of disagrees, we can search over all angles $\vec{\gamma}, \vec{\beta}$, or equivalently over all Laurent pairs $(A,B) \in \mathcal{S}_{2p-1}$.
\begin{corollary}
\label{cor:qaoa_ring_equals_expression}
Suppose $n = 2k$. For $p \ge 1$,
\begin{align*}
 \frac{1}{2k} \cdot \max_{\vec{\gamma}, \vec{\beta}} \bra{\vec{\gamma}, \vec{\beta}} H_C \ket{\vec{\gamma}, \vec{\beta}} 
 = \max_{(A,B) \in \mathcal{S}_{2p-1}} \left[ 1 -  \frac{1}{4k} \sum_{j=1}^k 
 \left| z_j^{-2} A(z_j)  + z_j^2 A^\#(z_j) - iB(z_j)  - i B^\#(z_j) \right|^2
 \right] \,,
 \end{align*}
 where
 $$
 z_j =e^{i \theta_j / 2} = \exp[ i \frac{\pi}{2} \cdot \frac{2j-1}{2k}]\,.
 $$
 \end{corollary}
 \begin{proof}
 Using \Cref{claim:fermion}, we have
 $$
 \bra{\vec{\gamma}, \vec{\beta}} H_C \ket{\vec{\gamma}, \vec{\beta}} = k - \sum_{j=1}^k  \bra{\chi(\theta_j)} h_c(\theta_j) \ket{\chi(\theta_j)} \,,\qquad \theta_j = \pi \cdot \frac{2j-1}{2k}\,.
 $$
 The corollary follows after using \Cref{thm:express} to expand each $\bra{\chi(\theta_j)} h_c(\theta_j) \ket{\chi(\theta_j)}$ as a Laurent polynomial in $z_j = e^{i \theta_j / 2}$.
 \end{proof}
Given a set of angles $\vec{\gamma}, \vec{\beta}$, one can construct the Laurent pair $(A,B)$ using \Cref{thm:state-theorem}. But the reverse direction is also possible: given the Laurent pair  $(A,B) \in \mathcal{S}_{2p-1}$, there are efficient algorithms to compute the angles $\phi_1, \dots, \phi_{2p}$~(e.g. \cite{Chao2020Angles, Ying2022Stable}), and thus the angles $\vec{\gamma}, \vec{\beta}$.

\section{When the QAOA cannot see the whole graph}
Let $F$ be the Laurent polynomial in $z$ that represents the amplitude $\braket{\theta}{\chi(\theta)}$:
\begin{equation}
F(z) := \frac{1}{2} \left( z^{-2} A(z) + z^2 A^\#(z) - iB(z) - i B^\#(z) \right) \,.
\end{equation}
Then the performance ratio of the QAOA on $C_{2k}$ is always of the form
$$
1 - \frac{1}{k} \sum_{j=1}^k |F(e^{i \theta_j/2)}|^2\,.
$$
In the depth-$p$ QAOA, $F$ is a Laurent polynomial of degree at most $2p+1$. Moreover, $F^\# = F$. So
$$
F(z) = \sum_{m=0}^{2p+1} c_{2p+1-2m} \cdot z^{2p+1-2m} = \sum_{m=0}^{p} c_{2m+1} (z^{2m+1} + z^{-(2m+1)}) \,.
$$
When the QAOA cannot see the whole graph, the performance ratio simplifies. First, consider $|F|^2$:
\begin{align*}
|F(e^{i\theta_j/2})|^2 &= \left|2 \sum_{m=0}^p c_{2m+1} \cos((2m+1) \cdot \theta_j /2)\right|^2 
\\
&= 4\sum_{m,m'=0}^p c_{2m+1} c_{2m'+1}^* 
\cos((2m+1) \cdot \theta_j /2)
\cos((2m'+1) \cdot \theta_j /2)
\\
&= 4\sum_{m,m'=0}^p \frac{1}{2} c_{2m+1} c_{2m'+1}^* \left( \cos(\theta_j \cdot (m - m')) + \cos(\theta_j \cdot (m + m' + 1)) \right)\,.
\end{align*}
The angles $\{\theta_1, \dots, \theta_k \}$ are equally spaced at $\pi/k$ increments around the half circle. As a result, many $\cos$ terms average to zero:
\begin{align*}
    \sum_{j=1}^k \cos(\theta_j \cdot c) = \sum_{j=1}^k \cos(\frac{2j-1}{2k} \cdot c \pi) = \mathbbm{1}_{c = 0\text{ mod } 2k} \cdot k\cos(\frac{c \pi}{2k}) 
\end{align*}
We track the nonzero terms when the QAOA cannot see the whole graph, i.e. $p < k$:
\begin{itemize}
    \item For the first term, this only occurs when $(m - m') = 0 \mod 2k$; this can only occur at $m = m'$.
    \item For the second term, this only occurs when $(m + m' + 1)= 0 \mod 2 k$. This never occurs when $p < k$, so this term contributes zero.
\end{itemize}
Altogether, when $p < k$, the performance ratio of the depth-$p$ QAOA takes the form
$$
1 - 2 \sum_{m=0}^p |c_{2m+1}|^2 = 1 - \sum_{m=0}^{2p+1}|c_{2p+1-2m}|^2 \,.
$$

\subsection{Upper bound}
We first use Cauchy-Schwarz to prove an upper bound on the QAOA's performance. There are other proofs of this fact  (e.g.~\cite{mbeng2019quantum}), but our proof will help us identify the lower bound.
\begin{claim}
\label{claim:qaoa_not_wholegraph_upper}
    Suppose $n = 2k$. For any $k > p \ge 0$, 
    $$
   \frac{1}{2k} \cdot \max_{\vec{\gamma}, \vec{\beta}} \bra{\vec{\gamma}, \vec{\beta}} H_C \ket{\vec{\gamma}, \vec{\beta}} \le \frac{2p+1}{2p+2}\,.
    $$
\end{claim}
\begin{proof}
The claim is trivial at $p = 0$, since 
$$
\frac{1}{2k} \cdot \max_{\vec{\gamma}, \vec{\beta}} \bra{\vec{\gamma}, \vec{\beta}} H_C \ket{\vec{\gamma}, \vec{\beta}} = \frac{1}{2k} \bra{+}^{\otimes 2k} H_C \ket{+}^{\otimes 2k} = \frac{1}{2}\,.
$$
We now study $p \ge 1$. On one qubit, the mixer Hamiltonian at $\theta = 0$ matches the cost Hamiltonian, and so the QAOA layers are rotations about the $X$-axis. Since the initial state is $\ket{+}$, the state remains $\ket{+}$ up to a global phase. Then the squared amplitude $|\braket{\theta}{\chi(\theta)}|^2 = |\braket{+}{+}|^2 = 1$. Since 
$$
|F(z)|^2 = |\braket{\theta}{\chi(\theta)}|^2\,,
$$
we have
$$
1 = |F(e^{i \cdot 0/2})|^2 = \left| \sum_{m=0}^{2p+1}c_{2p+1-2m} \right|^2 \le (2p+2) \sum_{m=0}^{2p+1} |c_{2p+1-2m}|^2\,,
$$
by Cauchy-Schwarz. Then
\[
    \sum_{m=0}^{2p+1} |c_{2p+1-2m}|^2 \ge\frac{1}{2p+2}\,.
\]
So the performance ratio of the depth-$p$ QAOA is
\begin{align*}
1 - \sum_{m=0}^{2p+1} |c_{2p+1-2m}|^2 \le 1 - \frac{1}{2p+2}\,.\tag*{\qedhere}
\end{align*}
\end{proof}

\subsection{Lower bound}
We now show that the QAOA can meet this upper bound. The only inequality in \Cref{claim:qaoa_not_wholegraph_upper} is Cauchy-Schwarz, which is saturated when $F$ has coefficients of equal value.
\begin{claim}
\label{claim:qaoa_not_whole_graph_lower}
     Suppose $n = 2k$. For any $k > p \ge 0$,
     \begin{align*}
    \frac{1}{2k} \cdot \max_{\vec{\gamma}, \vec{\beta}} \bra{\vec{\gamma}, \vec{\beta}} H_C \ket{\vec{\gamma}, \vec{\beta}} \ge \frac{2p+1}{2p+2}\,.
    \end{align*}
\end{claim}
\begin{proof}
The claim is trivial at $p = 0$, since 
$$
\frac{1}{2k} \cdot \max_{\vec{\gamma}, \vec{\beta}} \bra{\vec{\gamma}, \vec{\beta}} H_C \ket{\vec{\gamma}, \vec{\beta}} = \frac{1}{2k} \bra{+}^{\otimes 2k} H_C \ket{+}^{\otimes 2k} = \frac{1}{2}\,.
$$
We now study $p \ge 1$. Consider the polynomial
$$
F_p(z) = \frac{1}{2p+2} \sum_{m=0}^p (z^{2m+1} + z^{-(2m+1)}) = \frac{1}{2p+2} \sum_{m=0}^{2p+1} z^{2p+1-2m}\,.
$$
If $F = F_p$, the performance ratio of the depth-$p$ QAOA is
$$
1 -  \sum_{m=0}^{2p+1} |c_{2p+1-2m}|^2 = 1 - \frac{2p+2}{(2p+2)^2} = \frac{2p+1}{2p+2}\,.
$$
In our convention $z = e^{i \theta/2}$, $F_p$ has a simpler form:
$$
F_p(z)
= \frac{1}{2p+2}\frac{z^{2p+2} - z^{-2p-2}}{z - z^{-1}}
= \frac{1}{2p+2}\frac{\sin((p+1)\theta)}{\sin(\theta/2)}
= \frac{\cos(\theta/2)}{p+1} U_p(\cos \theta)\,,
$$
where $U_p$ is the $p^\text{th}$ Chebyshev polynomial of the second kind. $U_p$ has real coefficients, and satisfies
\begin{align*}
        &U_p(1) = p + 1\,, &  &U_p(1) > 0\,, & U_p(y) \ge U_p(1) \text{ if } y > 1\,, &  &|U_p(y)| \le U_p(1) \text{ if } {-1} \le y \le 1\,.
\end{align*}
Then there exists a polynomial pair $(A,B) \in \mathcal{S}_{2p-1}$ that sets $F$ equal to $F_p$.  The proof is somewhat technical; it is deferred to \Cref{claim:find_poly_pair_from_cheby} in \Cref{apx:optimal_poly}.
\end{proof}

\subsection{System size independence}
So far we have only looked at cycles with even length ($n = 2k$). We now handle the case when $n$ is odd. Recall that the QAOA is a local algorithm. This is because of a lightcone argument. The chance any edge is cut is a function only of its local neighborhood. When the QAOA cannot see the whole graph the local neighborhood around each edge is always the same regardless of $n$.
\begin{fact}[\cite{Farhi2014QAOA}]
\label{fact:scale_invariance}
   For any choice of $\vec{\gamma}, \vec{\beta}$, the value
    \begin{align*}
         \frac{1}{n} \cdot  \bra{\vec{\gamma}, \vec{\beta}} H_C \ket{\vec{\gamma}, \vec{\beta}} 
    \end{align*}
    is fixed for all $n \ge 2p + 2$.
\end{fact}
\Cref{fact:scale_invariance} implies that Claims \ref{claim:qaoa_not_wholegraph_upper} and \ref{claim:qaoa_not_whole_graph_lower} apply to cycles with odd length. This proves \Cref{thm:ring}.

\section{When the QAOA sees the whole graph}
\Cref{thm:ring2} follows by \Cref{claim:qaoa_whole_graph} (when $n$ is even) and \Cref{claim:qaoa_odd_whole_graph} (when $n$ is odd). We prove each case.
\subsection{When $n$ is even}
During the preparation of this document we found that \cite{Rabinovich2025Overparametrization} also proved \Cref{claim:qaoa_whole_graph}.
\begin{claim}
\label{claim:qaoa_whole_graph}
    Suppose $n = 2k$. For any $p \ge k \ge 2$,
   \begin{align*}
    \frac{1}{2k} \cdot \max_{\vec{\gamma}, \vec{\beta}} \bra{\vec{\gamma}, \vec{\beta}} H_C \ket{\vec{\gamma}, \vec{\beta}} = 1\,.
    \end{align*}
\end{claim}
\begin{proof}
    Fix $p$. Consider the Laurent polynomials
    \begin{equation}
  A_p(z)=\frac12\left(z^{3-2p}+z^{1-2p}\right)\,,
  \qquad
  B_p(z)=\frac12(z-z^{-1})\,.
  \label{eq:A-B-sparse-cos}
\end{equation}
We verify that $(A_p, B_p) \in \mathcal{S}_{2p-1}$. They have real coefficients, powers of $z$ with odd parity, and degree at most $2p-1$. It remains to check the normalization:
$$
A_p A_p^\# = \frac{1}{4} \left( 2 + z^2 + z^{-2} \right)\,, \qquad
B_p B_p^\# = -B_p^2 = \frac{1}{4} \left( 2 - z^2 - z^{-2} \right)\,.
$$
Then  $A_p A_p^\# + B_p B_p^\# = 1$. 

By \Cref{cor:qaoa_ring_equals_expression}, the performance ratio of the QAOA on $C_{2k}$ is at least
$$
 \frac{1}{2k} \cdot \max_{\vec{\gamma}, \vec{\beta}} \bra{\vec{\gamma}, \vec{\beta}} H_C \ket{\vec{\gamma}, \vec{\beta}} 
 \ge 1 -  \frac{1}{4k} \sum_{j=1}^k 
 \left| z_j^{-2} A_p(z_j)  + z_j^2 A_p^\#(z_j) - iB_p(z_j)  - i B_p^\#(z_j) \right|^2 \,,
$$
where $z_j = e^{i \theta_j / 2} = \exp[i \frac{\pi}{2} \cdot \frac{2j-1}{2k}]$. The below expression simplifies when $z = e^{i \theta/2}$:
\begin{align*}
F(z) &= \frac{1}{2} \left( z^{-2} A_p(z) + z^2 A_p^\#(z) - iB_p(z) - iB_p^\#(z) \right)
\\
&= \frac{1}{4} \left( z^{1-2p} + z^{-1-2p} + z^{2p-1} + z^{2p+1} - 2i B_p(z) + 2i B_p(z)\right)
\\
&= \frac{1}{4}(z+z^{-1})(z^{2p} + z^{-2p}) \\
&= \cos (\theta/2) \cos(p \theta)\,.
\end{align*}
So the performance ratio of the QAOA is at least
$$
1 - \frac{1}{k} \sum_{j=1}^k  \cos^2(\theta_j/2)\cos^2(p \theta_j) = 1 - \frac{1}{k} \sum_{j=1}^k \cos^2(\frac{\pi}{2} \cdot \frac{2j-1}{2k}) \cos^2(p \pi \cdot \frac{2j-1}{2k})\,. 
$$
Suppose $k  = p$. Then for all $1 \le j \le k$,
$$
\cos^2(p \pi \cdot \frac{2j-1}{2k}) = \cos^2(j \pi - \frac{\pi}{2}) = 0\,.
$$
So the performance ratio is exactly $1$. When $p \ge k$, we can use the associated angles for the first $k$ layers of the QAOA, and set the remaining angles to $0$.
\end{proof}
Since in expectation the QAOA finds cuts of maximum size there can be no fluctuation above this value. So the QAOA always finds a cut of maximum size. 
In this case, there is a simple choice of angles to optimize $C_{2k}$.
\begin{claim}
\label{claim:wholegraph_chosenangles}
    For any $p = k \ge 2$, the QAOA finds a perfect cut of $C_{2k}$ with angles 
    \begin{align*}
    \phi_j = \begin{cases}
    0 & 1 \le j \le p-1 \\
    \pi/4 & j = p \\
    -\pi/4 & j = p+1 \\
    0 & p+2 \le j \le 2p
    \end{cases}
    \end{align*}
    where 
    $$
    (-\gamma_1, 2\beta_1, \dots, -\gamma_p, 2\beta_p) = (\tilde{\gamma}_1, \tilde{\beta}_1, \dots, \tilde{\gamma}_p, \tilde{\beta}_p) = (\phi_1 + \frac{\pi}{2}, \phi_2 - \frac{\pi}{2}, \dots, \phi_{2p-1} + \frac{\pi}{2}, \phi_{2p} - \frac{\pi}{2})\,.
    $$
\end{claim}
\begin{proof}
Using \Cref{thm:state-theorem}, we decompose the Laurent pair chosen in \Cref{claim:qaoa_whole_graph}. Note that 
$$
\begin{pmatrix}
    z^{1-p} & 0 \\
    0 & z^{p-1}
\end{pmatrix}
\begin{pmatrix}
    (z + z^{-1})/2 & (z-z^{-1})/2 \\
    (z-z^{-1})/2 & (z + z^{-1})/2 \\
\end{pmatrix}
\begin{pmatrix}
    z^{1-p} & 0 \\
    0 & z^{p-1}
\end{pmatrix}
\begin{pmatrix}
1 \\ 
0
\end{pmatrix}
=
\frac{1}{2}
\begin{pmatrix}
    z^{3-2p} + z^{1-2p}\\
    z - z^{-1}
\end{pmatrix}
=
\begin{pmatrix}
    A_p(z)\\
   B_p(z)
\end{pmatrix}\,,
$$
Moreover,
$$
C(-\pi/4) D C(\pi/4) = 
\begin{pmatrix}
    1/\sqrt{2} & 1/\sqrt{2} \\
    -1/\sqrt{2} & 1/\sqrt{2}
\end{pmatrix}
\begin{pmatrix}
    z^{-1} & 0 \\
    0 & z
\end{pmatrix}
\begin{pmatrix}
    1/\sqrt{2} & -1/\sqrt{2} \\
    1/\sqrt{2} & 1/\sqrt{2}
\end{pmatrix}
=
\frac{1}{2}
\begin{pmatrix}
   z + z^{-1} & z - z^{-1} \\
   z - z^{-1} & z + z^{-1}
\end{pmatrix}\,.
$$
This is generated by $\phi_j = 0$ for $1 \le j \le p-1$, $\phi_p = \pi/4$, $\phi_{p+1} = -\pi/4$, and $\phi_j = 0$ for $p+2 \le j \le 2p$. 
So, the QAOA with these angles  (up to a factor of $2$ in $\beta$ and sign in $\gamma$) optimizes $C_{2k}$ when $p = k$.
\end{proof}
When $p > k$, the QAOA can optimize $C_{2k}$ starting with the angles in \Cref{claim:wholegraph_chosenangles}, and then setting $\gamma_j = \beta_j = 0$ for any $j > k$ (i.e. beyond depth $k$).

\subsection{When $n$ is odd}
When $n$ is odd the maximum cut size is $n-1$. It remains to show that the QAOA always finds such a cut when $n < 2p + 2$.

There is again a reduction to one-qubit systems when $n$ is odd:
\begin{claim}[\cite{Wang2018Fermionic}]
\label{claim:fermion_odd}
Consider the one-qubit Hamiltonian
$$
h_c(\theta) = X \cos\theta + Z \sin\theta
$$
and one-qubit QAOA state
$$
\ket{\chi(\theta)} := e^{-i \tilde{\beta}_p X}e^{-i \tilde{\gamma_p} h_c(\theta)} \dots e^{-i \tilde{\beta}_1 X}e^{-i \tilde{\gamma_1} h_c(\theta)} \ket{+}\,,
$$
where
$\vec{\tilde{\gamma}} := (\tilde{\gamma}_1, \dots, \tilde{\gamma}_p) = (-\gamma_1, \dots, -\gamma_p)$ and $\vec{\tilde{\beta}} := (\tilde{\beta}_1, \dots, \tilde{\beta}_p) = (2\beta_1, \dots, 2\beta_p)$.
For $n = 2k + 1$,
\begin{align*}
&\bra{\vec{\gamma}, \vec{\beta}} H_C \ket{\vec{\gamma}, \vec{\beta}} = k - \sum_{j=1}^k \bra{\chi(\theta_j)} h_c(\theta_j) \ket{\chi(\theta_j)} \,,\qquad \theta_j = \pi \cdot \frac{2j}{2k+1}\,.
\end{align*}
\end{claim}
Note that the angles $\theta_j$ are slightly different.  \Cref{claim:fermion_odd} allows us to use \Cref{thm:express}, which gives the following consequence:
\begin{corollary}
\label{cor:best_value_small_odd_cycle}
Suppose $n = 2k + 1$. For $p \ge 1$,
\begin{align*}
 \frac{1}{n} \cdot \max_{\vec{\gamma}, \vec{\beta}} \bra{\vec{\gamma}, \vec{\beta}} H_C \ket{\vec{\gamma}, \vec{\beta}} 
 = \max_{(A,B) \in \mathcal{S}_{2p-1}} \left[ \frac{n-1}{n} -  \frac{1}{2n} \sum_{j=1}^k 
 \left| z_j^{-2} A(z_j)  + z_j^2 A^\#(z_j) - iB(z_j)  - i B^\#(z_j) \right|^2
 \right] \,,
 \end{align*}
 where
 $$
 z_j =e^{i \theta_j / 2} = \exp[ i \frac{\pi}{2} \cdot \frac{2j}{2k+1}]\,.
 $$
\end{corollary}
We now prove that the QAOA always finds a cut of maximum size on the odd cycle.
\begin{claim}
\label{claim:qaoa_odd_whole_graph}
    Suppose $n = 2k+1$. For any $p \ge k \ge 1$,
   \begin{align*}
    \frac{1}{2k+1} \cdot \max_{\vec{\gamma}, \vec{\beta}} \bra{\vec{\gamma}, \vec{\beta}} H_C \ket{\vec{\gamma}, \vec{\beta}} = \frac{2k}{2k+1} = \frac{n-1}{n}\,.
    \end{align*}
\end{claim}
\begin{proof}
Since the maximum cut size is $n-1$, we just need to prove a lower bound.

Suppose the claim is true when $p = k$. Then when $p \ge k$, we can use the associated angles for the first $k$ layers of the QAOA, and set the remaining angles to $0$. So we focus on the case $p = k$. 

Invoking \Cref{cor:best_value_small_odd_cycle}, the claim is true if 
$$
F(z) = \frac{1}{2}  \left(  z^{-2} A(z)  + z^2 A^\#(z) - iB(z)  - i B^\#(z) \right) = 0
$$
for every $z_j = e^{i \theta_j / 2} = \exp[ i \frac{\pi}{2} \cdot \frac{2j}{2p+1}]$.

Consider the function
$$
F_p(z) = \frac{1}{4p+2}(z + z^{-1})  \sum_{m=-p}^p z^{2m}\,.
$$
For each $z_j$, the sum can be rewritten as 
$$
\sum_{m=-p}^p z_j^{2m} = \frac{z_j^{2p+1} - z_j^{-(2p+1)}}{z_j - z_j^{-1}} = \frac{\sin((2p+1)\theta_j/2)}{\sin(\theta_j/2)} = \frac{\sin(j\pi)}{\sin(j\pi/(2p+1))} = 0\,.
$$
So $F_p(z_j) = 0$ for each $z_j$. The claim is proven if there exists a Laurent polynomial pair $(A, B) \in \mathcal{S}_{2p-1}$ that sets $F = F_p$.

In our convention $z = e^{i \theta/2}$, $F_p$ has a simpler form:
\begin{align*}
F_p(z) &= \frac{1}{4p+2}(z + z^{-1}) \cdot \frac{z^{2p+2} + z^{2p} - z^{-2p} - z^{-(2p+2)}}{z^2 - z^{-2}}  \\
&= \frac{\cos(\theta/2)}{2p+1}  \cdot \frac{\sin((p+1)\theta) + \sin(p \theta)}{\sin(\theta)} \\
&= \frac{\cos(\theta/2)}{2p+1}
\left( U_{p}(\cos \theta) + U_{p-1}(\cos \theta)
\right)\,,
\end{align*}
where $U_{p}$ is the ${p}^\text{th}$ Chebyshev polynomial of the second kind. 

Note that the polynomial $V(y) = U_p(y) + U_{p-1}(y)$ has real coefficients, and satisfies
\begin{align*}
        &V(1) = 2p + 1\,, &  &V(1) > 0\,, & V(y) \ge V(1) \text{ if } y > 1\,, &  &|V(y)| \le V(1) \text{ if } {-1} \le y \le 1\,.
\end{align*}
In fact the last three conditions are preserved under positive linear combinations, so they are satisfied by positive sums of Chebyshev polynomials. By \Cref{claim:find_poly_pair_from_cheby}, there exists a polynomial pair $(A,B) \in \mathcal{S}_{2p-1}$ that sets $F$ equal to $F_p$. 
\end{proof}

\section{Discussion}
We have shown that the QAOA finds cuts with largest possible expected size among symmetric local algorithms on the ring of disagrees. 
It is surprising that we do not know any \textit{classical} algorithm in this family matching the performance of the QAOA.
Our results also apply when the edges of the ring are an assortment of agree and disagree, since the performance of the QAOA does not depend on the sign of any edge (e.g.~\cite[Appendix A]{mbeng2019quantum}).

Although we proved the existence of optimal angles, we do not know a closed form for these angles. They must be related to the roots of Chebyshev polynomials.

We are not aware of prior work studying the QAOA in the language of quantum signal processing. 
Another quantum algorithm which is optimized with a good choice of Laurent polynomials is~\cite{Farhi1999Insertion}.

During the preparation of this document we became aware of another group (Uri Kol, Maor Ben Shahar, Kfir Sulimany, Dirk Englund) independently proving \Cref{thm:ring} using the Lean proof assistant. Unlike this work, such a proof is written in computer code, and can be formally verified, once the definitions and theorem statements are properly checked.

This work was completed with the help of contemporary artificial intelligence systems. It is exciting to see what new research can be done by humans equipped with this emerging technology.

\section*{Acknowledgements}
Thanks to Edward Farhi and Sam Gutmann for introducing the problem to me.
Thanks to Richard Allen for teaching me about QSP, Sabee Grewal for teaching me about fermions, and Antonio Anna Mele for suggesting~\cite{mbeng2019quantum}. Thanks to Richard Allen, Ryan Babbush, Edward Farhi, Sam Gutmann, Zhang Jiang, and Benjamin Villalonga for comments on a draft of this document.

\section*{AI Disclosure}
The author extensively used ChatGPT 5.5 Pro and to some extent Claude Opus 4.8. I had several interactive sessions over the period of one week in early June 2026, often for many hours at a time.
I used both models to explain concepts and numerically investigate optimal angles on the ring.
I used ChatGPT 5.5 Pro to discover the connection to QSP, identify optimal Laurent polynomials, and collaboratively discover all proofs. ChatGPT 5.5 Pro wrote a proof of \Cref{thm:state-theorem}, which I read and polished. It generated many other ``proofs'' which were directionally correct; I used the output to prove these claims myself. Nearly every line in this document was written by a human.

\clearpage
\newpage
\printbibliography
\appendix
\clearpage
\newpage
\section{Proof of reduction to one-qubit systems (\texorpdfstring{\Cref{claim:fermion}}{Claim \ref{claim:fermion}})}
\label{apx:fermion}

For convenience, we rederive the result in \cite{Wang2018Fermionic} in detail. We focus on $n = 2k$. At the end we comment on how odd $n$ changes the picture.

 First, let $\tilde{H}_C = \sum_{i=1}^{2k} Z_i Z_{i+1}$. 
 All indices are modulo $2k$.
 Then 
$$
\bra{\vec{\gamma}, \vec{\beta}}H_C\ket{\vec{\gamma}, \vec{\beta}} = k - \frac{1}{2}\bra{\vec{\gamma}, \vec{\beta}}\tilde{H}_C\ket{\vec{\gamma}, \vec{\beta}} \,.
$$
We wish to minimize $\bra{\vec{\gamma}, \vec{\beta}}\tilde{H}_C\ket{\vec{\gamma}, \vec{\beta}}$. We then apply the Jordan-Wigner transformation. Let
\begin{align*}
    S_j^{\pm} &=\frac{1}{2} \left( Y_j \pm i Z_j \right) \\
    a_j &= S_j^{-} \prod_{j'=1}^{j-1} (-X_{j'})  \\
        a_j^\dagger &= S_j^{+} \prod_{j'=1}^{j-1} (-X_{j'})
\end{align*}
We verify the canonical anticommutation relations. For $v > u$, we have
\begin{align*}
a_v a_{u} &=  S_v^- \prod_{j=1}^{v-1} (-X_j) \cdot S_u^- \prod_{j=1}^{u-1} (-X_j) = - S_v^{-} S_u^{-} \prod_{j=u}^{v-1} (-X_j)\,, \\
a_u a_{v} &=  S_u^- \prod_{j=1}^{u-1} (-X_j) \cdot S_v^- \prod_{j=1}^{v-1} (-X_j) = S_u^{-} S_v^{-} \prod_{j=u}^{v-1} (-X_j)  = -a_v a_u\,,
\end{align*}
so $\{a_u, a_v\} = 0$. Similarly, $\{a_u^\dagger, a_v^\dagger\} = 0$ and $\{a_u, a_v^\dagger\} = \{a_v^\dagger, a_u\} = 0$.

When $u=v$, we have
\begin{align*}
a_u^2 &= (S_u^{-})^2 = \frac{1}{4} (Y_u-iZ_u)(Y_u-iZ_u) = \frac{1}{4} \left( Y_u^2 -iZ_uY_u -iY_uZ_u - Z_u^2 \right) = 0\,, \\
a_u a_u^\dagger &= S_u^{-} S_u^{+} = \frac{1}{4} (Y_u-iZ_u)(Y_u+iZ_u) =  \frac{1}{4} \left( Y_u^2 -iZ_uY_u + iY_uZ_u + Z_u^2 \right) = \frac{1}{2}\left( I-X_u\right)\,.
\end{align*}
so $\{a_u, a_u\} = 0$. Similarly, $\{a_u^\dagger, a_u^\dagger\} = 0$. However, $\{a_u, a_u^\dagger\} = \frac{1}{2}\left(I - X_u + I + X_u\right) = I$. 

In this picture, 
$$
H_M = \sum_{i=1}^{2k} (2 a_i^\dagger a_i - 1)
$$
For $u < 2k$, we have
$$
(a_u + a_{u}^\dagger)(a_{u+1} - a_{u+1}^\dagger) = Y_u (-iZ_{u+1}) (-X_u) = Z_u Z_{u+1}\,. 
$$
For $u = 2k$, we have
$$
(a_{2k} + a_{2k}^\dagger)(a_1 - a_1^\dagger) = - (-iZ_1)Y_{2k} \prod_{i=1}^{2k-1} (-X_i) = - Z_1 Z_{2k} \prod_{i=1}^{2k} (-X_i)\,.
$$
Altogether, this means
$$
\tilde{H}_C = - (a_{2k} + a_{2k}^\dagger)(a_1 - a_1^\dagger) P_X + \sum_{u=1}^{2k-1} (a_u + a_{u}^\dagger)(a_{u+1} - a_{u+1}^\dagger)
$$
where $P_X = \prod_{i=1}^{2k} X_i$.

We then add a phase factor to unify the expression. Let $b_j := a_j e^{-i \pi  \cdot j/(2k)}$. Then 
$$
(e^{i \pi \cdot j/(2k)} b_j + e^{-i \pi \cdot j/(2k)} b_j^\dagger)(e^{i \pi \cdot (j+1)/(2k)} b_{j+1} - e^{-i \pi \cdot (j+1)/(2k)} b_{j+1}^\dagger) 
= e^{-i \pi/2k} \left( e^{i \pi (j+1)/k} b_j b_{j+1} - b_j b^\dagger_{j+1} \right) + \text{h.c.} 
$$
Similarly,
$$
- (a_{2k} + a_{2k}^\dagger)(a_1 - a_1^\dagger) = (b_{2k} + b_{2k}^\dagger)(e^{i \pi / 2k} b_1 - e^{-i\pi/2k} b_1^\dagger)  
= e^{-i \pi/2k} \left( e^{i \pi/k} b_{2k} b_1 - b_{2k} b_1^\dagger \right) + \text{h.c.}
$$
Moreover, the QAOA stays within the $P_X = 1$ subspace, as we start in $\ket{+}^{\otimes 2k}$, and $P_X$ commutes with $H_M$ and $H_C$. So within this subspace,
\begin{align*}
    H_M &= \sum_{j=1}^{2k} (2 b_j^\dagger b_j - 1)\,, \\
    \tilde{H}_C &= e^{-i\pi/2k} \sum_{j=1}^{2k} 
    \left( e^{i \pi (j+1)/k} b_j b_{j+1} - b_j b^\dagger_{j+1} \right) 
    +
    e^{i \pi/2k} \sum_{j=1}^{2k}
    \left(  b_j^\dagger b_{j+1} - e^{-i \pi (j+1)/k} b_j^\dagger b^\dagger_{j+1} \right) \,.
\end{align*}
We now apply a Fourier transform. Let $c_m = \frac{1}{\sqrt{2k}} \sum_{j=1}^{2k}e ^{i\omega \cdot jm} b_j$, where $\omega = 2 \pi / (2k)$. Then
\begin{align*}
\sum_{m=1}^{2k} e^{-i\omega m} c_m^\dagger c_m &= \frac{1}{2k} \sum_{m=1}^{2k} \sum_{j,j'=1}^{2k} e^{i\omega (-1-j+j')m} b_j^\dagger b_{j'} 
= \frac{1}{2k} \sum_{j,j'=1}^{2k} 2k \cdot \delta_{j'=1+j} b_j^\dagger b_{j'} 
= \sum_{j=1}^{2k} b_j^\dagger b_{j+1}\,, \\
    \sum_{m=1}^{2k} e^{i\omega m} c_{m} c_{1-m} 
    &= \frac{1}{2k} \sum_{m=1}^{2k} \sum_{j,j'=1}^{2k} e^{i\omega (1+j-j')m} e^{i\omega j'} b_j b_{j'}
    = \frac{1}{2k} \sum_{j,j'=1}^{2k} 2k \cdot \delta_{j'=1+j} e^{i\omega j'} b_j b_{j'} 
    = \sum_{j=1}^{2k} b_j b_{j+1} e^{i\omega (j+1)}\,.
\end{align*}
So the system becomes organized in pairs of two fermions:
\begin{align*}
        H_M &= \sum_{j=1}^{2k} (2 c_j^\dagger c_j - 1)\,, \\
        \tilde{H}_C &= e^{-i\omega/2} \sum_{m=1}^{2k} e^{i\omega m} c_m (c_{1-m} - c_m^\dagger) + e^{i\omega/2} \sum_{m=1}^{2k} e^{-i\omega m} c_m^\dagger (c_m - c_{1-m}^\dagger).
\end{align*}
We group the cost Hamiltonian. Let $\theta_m = (m-1/2) \omega$. Then $\theta_{1-m} = - \theta_m$. So
\begin{align*}
\tilde{H}_C &= \sum_{m=1}^{2k} e^{i \theta_m} \left(  c_m c_{1-m} - c_m c_m^\dagger \right) + e^{-i\theta_m}\left( c_m^\dagger c_m - c_m^\dagger c_{1-m}^\dagger \right)
\end{align*}
We simplify this expression.
Since $\{c_m, c_{1-m}\} = 0$, we have $e^{i \theta_m} c_m c_{1-m} + e^{i \theta_{1-m}} c_{1-m} c_{1-(1-m)} = 2 i \sin(\theta_m) c_m c_{1-m}$. The same holds for $\{c_m^\dagger, c_{1-m}^\dagger\}$. Moreover,
$$
-e^{i \theta_m} c_m c_m^\dagger + e^{-i \theta_m} c_m^\dagger c_m = 
-e^{i \theta_m} (I - c_m^\dagger c_m) + e^{-i \theta_m} c_m^\dagger c_m = -e^{i \theta_m } +  2\cos(\theta_m) c_m^\dagger  c_m\,.
$$
Note that $\cos(\theta_m) = \cos(\theta_{1-m})$, and $e^{i \theta_m} + e^{i \theta_{1-m}} = 2 \cos(\theta_m)$. So 
$$
\tilde{H}_C = \sum_{m=1}^{k} 2i \sin(\theta_m) (c_m c_{1-m} + c_m^\dagger c_{1-m}^\dagger)  + 2 \cos(\theta_m)(c_m^\dagger c_m + c_{1-m}^\dagger c_{1-m} - 1)\,.
$$
We start in the maximum eigenstate of $H_M$, which is the maximum occupancy state for all fermions $c_j$. Both $H_M$ and $\tilde{H}_C$ preserve the parity occupancy of modes $\{m, 1-m\}$; i.e. fermions will occupy both modes or neither mode. So we can treat each pair of fermions as a single ``pseudospin''. For example, $c_m^\dagger c_m + c_{1-m}^\dagger c_{1-m}-1$ checks for occupancy of the fermion pair, i.e. acts as $\ket{11}\bra{11} - \ket{00}\bra{00}$. Meanwhile, $c_m c_{1-m} + c^\dagger_m c^\dagger_{1-m}$ flips the occupancy, i.e. couples $\ket{00}$ and $\ket{11}$. We consider this ``pseudospin'' in the $X$ basis; i.e. $X_m := c_m^\dagger c_m + c_{1-m}^\dagger c_{1-m} - 1$, and $Z_m = i(c_m c_{1-m} + c_m^\dagger c_{1-m}^\dagger)$. Then
\begin{align*}
    H_M &= 2\sum_{m=1}^k X_m \\
    \tilde{H}_C &= 2 \sum_{m=1}^k \cos(\theta_m) X_m + \sin(\theta_m) Z_m \quad &\theta_m = \pi \cdot (2m-1)/(2k)\,.
\end{align*}
In this picture, the initial state is $\ket{+}^{\otimes k}$. Since both $\tilde{H}_C$ and $H_M$ are $1$-local, the final state is a product state. So we can look at each qubit independently. Given
\begin{align*}
h_c(\theta) &:= X \cos \theta + Z \sin \theta\,, \\
\ket{\chi(\theta)} &:= e^{-i 2\beta_p X}e^{i \gamma_p h_c(\theta)} \dots e^{-i 2\beta_1 X}e^{i \gamma_1 h_c(\theta)} \ket{+}\,,
\end{align*}
we can rewrite the performance as
$$
\bra{\vec{\gamma}, \vec{\beta}} H_C \ket{\vec{\gamma}, \vec{\beta}} = k - \frac{1}{2} \bra{\vec{\gamma}, \vec{\beta}} \tilde{H}_C\ket{\vec{\gamma}, \vec{\beta}} = k - \sum_{j=1}^k \bra{\chi(\theta_j)} h_c(\theta_j) \ket{\chi(\theta_j)}\,.
$$
The $\beta$ angles are larger by a factor of two since after we grouped $2k$ sites to $k$ sites the mixer $H_M$ doubles in strength. The $\gamma$ angles change sign, since $H_C = k - \sum_{j=1}^k h_c(\theta_j)$.

\paragraph{Odd $n$}
Suppose $n = 2k + 1$. When $n$ is odd a few things change. We do not need to move from $a_j$ to $b_j$ since the QAOA stays within the $\prod_{i=1}^n(-X_i) = -1$ subspace. Instead of $(c_m, c_{1-m})$, the fermions $(c_m, c_{-m})$ are paired for $1 \le m \le k$. This changes the definition of angle $\theta_m$ to be $\theta_m = m \omega = 2\pi \cdot m/n$. For $c_0$ the fermion is unpaired, but since $\theta_0 = 0$ the operator in $\tilde{H}_C$ matches that in $H_M$. So the one-qubit system corresponding to $c_0$ always gives energy $1$ according to $\tilde{H}_C$. This makes the constant term of $H_C$  equal to $\frac{2k+1}{2} - \frac{1}{2} = k$.
\clearpage
\newpage
\section{Proof of statements from quantum signal processing}
\label{apx:qsp}

\subsection{Proof of \texorpdfstring{\Cref{thm:state-theorem}}{Claim \ref{thm:state-theorem}}}
We start by showing the first set is contained in the second set. The starting vector $(1,0)^T$ corresponds to $(A,B)=(1,0)\in\mathcal{S}_0$. The two elementary actions are
\[
  C(\phi)\begin{pmatrix}A\\B\end{pmatrix}
  =\begin{pmatrix}\cos\phi\,A-\sin\phi\,B\\
  \sin\phi\,A+\cos\phi\,B\end{pmatrix}\,,
\]
and
\[
  D\begin{pmatrix}A\\B\end{pmatrix}
  =\begin{pmatrix}z^{-1}A\\zB\end{pmatrix}\,.
\]
Each application of $D$ changes the parity of the exponents, and increases the Laurent degree by at most $1$. The identity
\[
  AA^\#+BB^\#=1
\]
is preserved on the unit circle $|z| = 1$ because there $C(\phi)$ and $D$ are unitary. Since this is a polynomial identity true on an infinite number of points, it must be true for all $z$.

To show the second set is contained in the first set, we proceed by induction on $r$. The base case $r=0$ is immediate:
If $A,B\in\R$ and $A^2+B^2=1$, then $(A,B)^T=C(\phi_0)(1,0)^T$ for some $\phi_0$.

Now we assume the theorem holds for $r-1$ and take $(A,B)\in\mathcal{S}_r$. Write
\[
  A(z)=\sum_{j=-r}^r a_jz^j,
  \qquad
  B(z)=\sum_{j=-r}^r b_jz^j.
\]
The coefficient of $z^{2r}$ in $AA^\#+BB^\#$ is
\[
  a_ra_{-r}+b_rb_{-r}\,.
\]
Since $AA^\#+BB^\#=1$, this coefficient is zero. We choose $\phi$ so that, after applying $C(-\phi)$,
\[
  \begin{pmatrix}A'\\B'\end{pmatrix}:=C(-\phi)\begin{pmatrix}A\\B\end{pmatrix},
\]
the coefficient of $z^{r}$ in $A'$ is equal to $0$. Assume $(a_r, b_r) \ne (0,0)$. For the first row, choose $(\cos \phi, \sin \phi)$ orthogonal to $(a_{r}, b_{r})$. The second row $(-\sin \phi, \cos \phi)$ will then be proportional to $(a_r, b_r)$, and thus orthogonal to $(a_{-r}, b_{-r})$. Then the coefficient of $z^{-r}$ in $B'$ is equal to $0$. If $(a_r, b_r) = (0,0)$ and $(a_{-r}, b_{-r}) \ne 0$, choose the first row proportional to $(a_{-r}, b_{-r})$. Otherwise use any $\phi$.

Now set
\[
  \begin{pmatrix}\widetilde A\\\widetilde B\end{pmatrix}
  :=D^{-1}\begin{pmatrix}A'\\B'\end{pmatrix}
  =\begin{pmatrix}zA'\\z^{-1}B'\end{pmatrix}.
\]
The terms in $A'$ are $z^{q}$ for $-r \le q \le (r-2)$, because of the parity constraint. So $zA'$ has degree at most $r-1$ and parity $r-1$.
Similarly, the terms in $B'$ are $z^{q}$ for $-r+2 \le q \le r$ because of the parity constraint.
So $z^{-1}B'$ has degree at most $r-1$ and parity $r-1$. Normalization is preserved: 
$$
\widetilde{A} \widetilde{A}^\# + \widetilde{B} \widetilde{B}^\# = A'(A')^\# + B'(B')^\# = A A^\# + B B^\# = 1\,.
$$
Thus $(\widetilde A,\widetilde B)\in\mathcal{S}_{r-1}$.  By the
induction hypothesis
\[
  \begin{pmatrix}\widetilde A\\\widetilde B\end{pmatrix}
  =C(\phi_{r-1})D\cdots D C(\phi_0)\begin{pmatrix}1\\0\end{pmatrix}.
\]
Therefore
\[
  \begin{pmatrix}A\\B\end{pmatrix}
  =C(\phi_r)D
  C(\phi_{r-1})D\cdots D C(\phi_0)\begin{pmatrix}1\\0\end{pmatrix},
\]
which completes the induction.

\subsection{Proof of \texorpdfstring{\Cref{cor:state-theorem}}{Corollary \ref{cor:state-theorem}}}
We start with the following lemma:
\begin{lemma}
If
\[
  V_r(z)=C(\phi_r)D C(\phi_{r-1})D\cdots D C(\phi_0)
\]
and
\[
  V_r(z)\begin{pmatrix}1\\0\end{pmatrix}
  =\begin{pmatrix}A(z)\\B(z)\end{pmatrix},
\]
then
\begin{equation}
  V_r(z)=
  \begin{pmatrix}
    A(z)&-B^\#(z)\\
    B(z)&A^\#(z)
  \end{pmatrix}.
  \label{eq:matrix-form-V}
\end{equation}
\end{lemma}
\begin{proof}
Let $V_r$ have the form
$$
V_r(z)= \begin{pmatrix}
    A(z) & G(z) \\
    B(z) & H(z)
\end{pmatrix}\,.
$$
By \Cref{thm:state-theorem}, the first column of $V_r$ is a Laurent pair in $\mathcal{S}_r$; i.e. $(A,B) \in \mathcal{S}_r$. By symmetry, the second column of $V_r$ is also a Laurent pair, i.e. $(G,H) \in \mathcal{S}_r$. 
On the unit circle, the matrix is unitary and has determinant $1$, so the following conditions hold:
\begin{align*}
    AH - BG &= 1 \\
    AG^* +  B H^* &= 0
\end{align*}
The linear equations above force $(G,H)$ to be orthogonal to $(A,B)$ and have norm 1. So $(G,H)$ must be equal to $(-B^*, A^*) = (-B^\#, A^\#)$. Since this holds on an infinite number of points, it is true for all $z$. 
\end{proof}
We then expand $V_r D$ in the $Y$-basis:
$$
V_r D =  \begin{pmatrix}
    A(z)&-B^\#(z)\\
    B(z)&A^\#(z)
  \end{pmatrix}  \begin{pmatrix}
    z^{-1} & 0 \\
    0 & z 
  \end{pmatrix} 
  =
  \begin{pmatrix}
    z^{-1} A(z)&-zB^\#(z)\\
    z^{-1} B(z)&zA^\#(z)
  \end{pmatrix}\,.
$$
\Cref{cor:state-theorem} follows by setting $r = 2p-1$.

\clearpage
\newpage
\section{Proof of existence of optimal Laurent polynomials}
\label{apx:optimal_poly}
To start, we require the following statement:
\begin{lemma}
\label{lem:roots_generic}
Consider a polynomial $V(y)$ with real coefficients that satisfies
\begin{align*}
        &V(1) > 0\,, & V(y) \ge V(1) \text{ if } y > 1\,, &  &|V(y)| \le V(1) \text{ if } {-1} \le y \le 1\,.
\end{align*}
Then the polynomial
$$
f(y) :=  -1 + (y/4 + 1) \frac{V^2(y/2 + 1)}{V^2(1)}
$$
has a simple root at $0$, and all other roots are complex conjugate pairs (i.e. not on the real line).
\end{lemma}
\begin{proof}
At $y = 0$, we have 
$
f(0) = -1 + \frac{V^2(1)}{V^2(1)} = 0
$.
We then consider
$$
\frac{f(y)}{y} = -\frac{1}{y} + (1/y + 1/4) \frac{V^2(y/2 + 1)}{V^2(1)}\,.
$$
We show $\frac{f(y)}{y}$ has no roots on the real line. Then it will have only complex roots; since all coefficients are real, these complex roots come in conjugate pairs. We split into four cases:
\begin{itemize}
    \item When $y < -4$, both terms are positive, so $f(y)/y > 0$. 
    \item When  $-4 \le y < 0$,  we have $|y/2 + 1| \le 1$, and so $|V(y/2 + 1)| \le |V(1)|$.  This implies $f(y)/y \ge -1/y + (1/y + 1/4) = 1/4 > 0$.
    \item When $y = 0$, $f(y)/y$ has a removable singularity, which is at least $1/4$ by continuity. 
    \item When $y > 0$, we have $V(y/2 + 1) \ge V(1)$. This implies $f(y)/y \ge 1/4$. \qedhere
\end{itemize}
\end{proof}
\begin{claim}
\label{claim:find_poly_pair_from_cheby}
    Fix $p \ge 1$. Consider a degree-$p$ polynomial $V(y)$ with real coefficients that satisfies
\begin{align*}
        &V(1) > 0\,, & V(y) \ge V(1) \text{ if } y > 1\,, &  &|V(y)| \le V(1) \text{ if } {-1} \le y \le 1\,.
\end{align*}
Then there exists a polynomial pair $(A,B) \in \mathcal{S}_{2p-1}$ that satisfies $B^\#= -B$ and
\begin{align*}
    F(z) &= \frac{1}{2} \left( z^{-2} A(z) + z^2 A^\#(z) - iB(z) - i B^\#(z) \right)  = (z + z^{-1}) \frac{V(z^2/2 + z^{-2}/2)}{2 V(1)}\,.
\end{align*}
\end{claim}
\begin{proof}
By \Cref{lem:roots_generic}, the polynomial
$$
f(y) :=  -1 + (y/4 + 1) \frac{V^2(y/2 + 1)}{V^2(1)}\,.
$$
has roots of the form $\{0, a_1, a_1^*, \dots, a_p, a_p^*\}$.
Let $c$ be the leading coefficient of $\frac{V(y/2)}{2V(1)}$.
We decompose
\begin{align*}
f(y) &= y q(y) q^*(y) \,,\\
    q(y) &:=   -c \prod_{m=1}^p \left(y - a_m \right) = w(y)  + ib(y) \,, \\
     q^*(y) &:=   -c \prod_{m=1}^p \left(y - a_m^* \right) = w(y) - ib(y) \,,
\end{align*}
where $w$ and $b$ have real coefficients. We then choose
\begin{align*}
&A(z) :=  z^2  (z - z^{-1}) \cdot w\left( (z - z^{-1})^2 \right) + z^2 (z + z^{-1}) \frac{V(z^2/2 + z^{-2}/2)}{2 V(1)}\,,  \\
&B(z) := (z - z^{-1})  \cdot b\left((z - z^{-1})^2\right)\,.
\end{align*}
We verify that $(A, B) \in \mathcal{S}_{2p-1}$.
All coefficients are real, and all powers of $z$ have odd parity. We defer checking the degree to the end. Here we check the normalization. Note that $B$ is $(z-z^{-1})$ times a function of $(z - z^{-1})^2$, so $B^\# = -B$. Next, $z^{-2} A$ has two parts: when exchanging $z$ with $z^{-1}$, the first part is odd  and the second part is even. So 
$$
A^\# A = (z^{-2} A) (z^{-2} A)^\# = -(z - z^{-1})^2 \cdot w^2\left( (z - z^{-1})^2 \right) + (z + z^{-1})^2 \frac{V^2(z^2/2 + z^{-2}/2)}{4 V^2(1)}\,.
$$
Altogether, 
\begin{align*}
A^\# A + B^\# B 
&= - (z - z^{-1})^2 \Big(w^2\left( (z - z^{-1})^2 \right) + b^2\left( (z - z^{-1})^2 \right) \Big) +(z + z^{-1})^2 \frac{V^2(z^2/2 + z^{-2}/2)}{4 V^2(1)} \\
&= -y (w(y) + ib(y))(w(y)-ib(y))
+ (y+4) \frac{V^2(y/2 + 1)}{4 V^2(1)} \\
&= (y/4 + 1) \frac{V^2(y/2 + 1)}{V^2(1)} - f(y)\,,
\end{align*}
where $y := (z - z^{-1})^2$. So the normalization is exactly $1$. 

We now inspect the degree of $A$ and $B$. Since the leading coefficient of $q$ is real, $b(y)$ has (ordinary) degree $\le (p-1)$. So $B(z)$ has Laurent degree at most $2p-1$. Meanwhile, $w(y)$ has ordinary degree $\le p$, so the Laurent exponents of $A(z)$ range from $1-2p$ through $2p + 3$. We confirm the coefficients at degree $2p+1$ and $2p+3$ are zero. The degree-$(2p+3)$ coefficient of $A$ is the sum of $-c$ (from the odd part) and $c$ (from the even part). So it is zero.

Now we check the degree-$(2p+1)$ coefficient of $A$. 
Since the normalization $A^\# A + B^\# B = 1$, 
$$
(z - z^{-1})^2 w^2(y) + (z - z^{-1})^2  b^2(y) =-1 +  (z + z^{-1})^2\frac{V^2(z^2/2 + z^{-2}/2)}{4V^2(1)}\,.
$$
Since $(z - z^{-1})^2 b^2(y)$ has degree $(4p-2)$ in $z$, the degree-$4p$ coefficient in $z$ only comes from $(z - z^{-1})^2  w^2(y)$ and the right-hand side.  On the right-hand side, it must be $2c c'$, where $c'$ is the next-largest coefficient of $(z+z^{-1})\frac{V(z^2/2 + z^{-2}/2)}{2V(1)}$. On the left-hand side, it is 
$$
2 d_{2p+1} d_{2p-1}\,,
$$
where 
$d_{j}$ is the degree-$j$ coefficient of $(z - z^{-1}) w(y)$.
We have already proved that $d_{2p+1} = -c \ne 0$, so it must be that $d_{2p-1} = -c'$. By inspecting the expression for $A$, the degree-$(2p+1)$ coefficient of $A(z)$ must also be zero. So $A$ has Laurent degree at most $2p-1$.

Finally, we check that $(A,B)$ generate the correct $F$. Recall that $B^\# = -B$ and the first part of $z^{-2} A$ is also odd when exchanging $z$ with $z^{-1}$. Then we get
\begin{align*}
F(z) &= \frac{1}{2} \left( z^{-2} A(z) + z^2 A^\#(z) - i B(z) - iB^\#(z)  \right)  \\
&= \frac{1}{2} \left( z^{-2} A(z)  + \left(z^{-2} A(z)\right)^\# \right) \\
&=
 (z + z^{-1}) \frac{V(z^2/2 + z^{-2}/2)}{2 V(1)}\,. \tag*{\qedhere}
\end{align*}
\end{proof}
\end{document}